\documentclass[11pt,a4paper]{article}
\usepackage{amsfonts,amsthm,amssymb,latexsym,amsmath,epsfig,color}
\usepackage{verbatim,graphicx,hyperref,shapepar}
\usepackage{dsfont}
\usepackage{a4wide}
\usepackage[utf8]{inputenc}
\usepackage{tikz}

\newcommand\job[3]{
    \draw[fill=black!#3,draw] (#1,0) -- (#1,#2) -- (#1+#2,#2) -- (#1,0);
    \node at (#1+#2*0.26,2*#2/3) {#2};
}
\newcommand\joblabel[3]{
    \draw[fill=white,draw] (#1,0) -- (#1,#2) -- (#1+#2,#2) -- (#1,0);
    \node at (#1+#2*0.26,2*#2/3) {#3};
}
\newcommand\jobdiscrete[3]{
    \foreach \i in {1,...,#2}
    {
        \draw[dotted] (#1,\i-1) rectangle (#1+\i,\i);
    }
    \draw (#1,0) -- (#1,#2) -- (#1+#2,#2);
    \foreach \i in {1,...,#2}
    {
        \draw (#1+\i-1,\i-1) -- (#1+\i,\i-1) -- (#1+\i,\i);
    }
}
\usetikzlibrary{arrows}

\newcommand{\mypara}[1]{\smallskip\noindent\textbf{#1.}}  

\theoremstyle{plain}
\newtheorem{definition}{Definition}
\newtheorem{theorem}{Theorem}

\newtheorem{lemma}{Lemma}
\newtheorem{claim}{Claim}

\title{The triangle scheduling problem}

\author{Christoph D\"urr%
\thanks{Sorbonne Universités, UPMC Univ Paris 06, CNRS, LIP6, Paris, France}
\and
Zden\v{e}k Hanz\'alek%
\thanks{Department of Control Engineering, The Czech Technical University in Prague, Czech Republic}
\and
Christian Konrad\thanks{School of Computer Science, Reykjavik University, Island}
\and
Yasmina Seddik\footnotemark[2]
\and
Ren\'e Sitters%
\thanks{Department of Econometrics and Operations Research, Vrije Universiteit, The Netherlands}
\and
\'Oscar C. V\'asquez%
\thanks{Department of Industrial Engineering, University of Santiago of Chile}
\and
Gerhard Woeginger%
\thanks{Department of Mathematics and Computer Science, Eindhoven University of Technology, The Netherlands}
}

\begin{document}

\maketitle

\begin{abstract}
  This paper introduces a novel scheduling problem, where jobs occupy a triangular shape on the time line.  This problem is motivated by scheduling jobs with different criticality levels.  A measure is introduced, namely the \emph{binary tree ratio}. It is shown that the greedy algorithm solves the problem to optimality when the binary tree ratio of the input instance is at most $2$.  We also show that the problem is unary NP-hard for instances with binary tree ratio strictly larger than 2, and provide a quasi polynomial time approximation scheme (QPTAS). The approximation ratio of Greedy on general instances is shown to be between 1.5 and 1.05.
\end{abstract}

\section{Introduction}

\mypara{Mixed-criticality Scheduling} In a mixed-criticality system, tasks with different criticality levels
coexist and need to share common resources, such as bandwidth in a communication channel \cite{hts16,hp98}
or execution time on a machine \cite{ves07,bbdlmms10,bls10,cnqw13,pk11,spbb13}.
Contrary to single-criticality systems, the estimated worst-case executing time (WCET) of a task depends on
its criticality level (the higher the criticality level, the more time is estimated). Often, however, the actual
execution times of tasks are not known beforehand and the estimated WCETs deviate hugely from the actual
execution times. The goal is therefore to design {\em robust} schedules that are able to tolerate runtime
variations to a reasonable extent. More conservative WCET estimates are usually used for highly critical tasks
(e.g. braking in a car), while less conservative estimates suffice for low-critical tasks (e.g. displaying the
temperature in a car). In case the allocated time for a task is insufficient at runtime, i.e.,
the actual runtime of a task exceeds its estimated WCET, the execution of the task may nevertheless continue
and suppress subsequent tasks of lower criticality. 
For a recent thorough survey on mixed-criticality systems and arising scheduling problems we refer the reader to \cite{bur14}.

\mypara{The Triangle Scheduling Problem}
In this paper, we consider the problem of non-preemptively scheduling $n$ unit length jobs/tasks with different
criticality levels on a single machine. Let $p_i \in \mathbb{N}$ denote the criticality level
of job $i$. The expected execution time of every job is $1$. If however at runtime, a job $i$ requires more time,
then we continue its execution for at most $p_i$ time units, and other jobs with lower criticality levels that were
scheduled in these slots are canceled. The computed schedule has to fulfill the property that, when
prolonging the execution of a job, no other job with larger criticality needs to be canceled. The objective
is to minimize the {\em makespan} (total completion time).

The previous problem can be see as a one-dimensional triangle alignment (or scheduling) problem, defined as follows:

\begin{definition}[Triangle scheduling, Gap] \label{def:ts}
Given integers $p_1, \ldots, p_n$ with $p_1\geq \ldots \geq p_n>0$ find starting times $s_1, \ldots, s_n \geqslant 0$
minimizing the so-called \emph{makespan} $\max_j s_j + p_j$ such that for all $i \neq j$ we have
$| s_i - s_j | \geqslant \min \{ p_i, p_j \}$.
We call \emph{gap} any interval spanned by successive starting times, including the interval between
the maximum starting time and the makespan of the schedule.
\end{definition}
We abbreviate the triangle scheduling problem by \textsc{TS}.

\begin{figure}[ht]
\begin{center}
\begin{tikzpicture}[scale=0.35,thick]
\jobdiscrete{0}{6}{0}
\jobdiscrete{4}{4}{0}
\jobdiscrete{7}{3}{0}
\jobdiscrete{10}{5}{0}
\draw[->](0,0) --  (15,0) node[below] {time};

\draw (0,-2) rectangle (1,-3);
\draw (4,-2) rectangle (5,-3);
\draw (7,-2) rectangle (8,-3);
\draw (10,-2) rectangle (11,-3);
\draw[->](0,-3) --  (15,-3);

\draw (0,-4) rectangle (5,-5);
\draw (7,-4) rectangle (9,-5);
\draw (10,-4) rectangle (14,-5);
\draw[->](0,-5) --  (15,-5);

\draw (0,-6) rectangle (4,-7);
\draw (4,-6) rectangle (8,-7);
\draw (10,-6) rectangle (12,-7);
\draw[->](0,-7) --  (15,-7);

\begin{scope}[shift={(20,0)}]
\job{0}{6}{0}
\job{4}{4}{0}
\job{7}{3}{0}
\job{10}{5}{0}
\draw[->](0,0) --  (15,0) node[below] {time};
\end{scope}
\end{tikzpicture}
\end{center}
  \caption{Example of an optimal schedule. \textbf{Left:} discrete version. \textbf{Below:} different possible executions.  In the first one each job requires a single time slot and all jobs are executed.  In the second and third executions, jobs require longer processing times preventing the execution of some jobs with lower criticality.
  \textbf{Right:} continuous version. Each triangle
  corresponds to a job $j$ labeled with its criticality level $p_j$. The left border of
  a triangle is the starting time of the job, while the right end is its worst-case completion time.}
    \label{fig:ExampleSolution}
\end{figure}
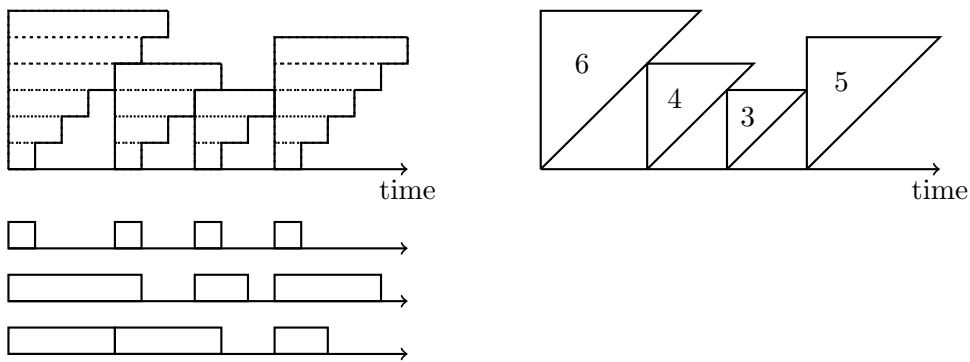


\mypara{Our Results}
In this paper, we initiate the study of \textsc{TS}. We show that the problem is strongly NP-hard,
which implies that \textsc{TS} does not admit a fully polynomial time approximation scheme (FPTAS),
unless $P=NP$. We provide a quasi polynomial time approximation scheme (QPTAS),
which implies that \textsc{TS} is not APX-hard.
In addition we present a Greedy algorithm, which processes the triangles from largest to smallest,
placing each into the largest gap and potentially shifting subsequent triangles to the right
if the gap is not large enough to contain it. We show that this algorithm has an approximation factor
between $1.05$ and $1.5$.

\mypara{Binary Tree Ratio} Furthermore, we establish a measure, denoted the {\em binary tree ratio}, that allows us to
distinguish hard from easy instances:
\begin{definition}[Binary tree ratio]
Given an instance of \textsc{TS} $p=p_1,\ldots,p_n$, we define its \emph{binary tree ratio} $R(p)$ as
\[
R(p) := \max_{i=2, \hdots,n} p_{\lceil i/2 \rceil} / p_i.
\]
\end{definition}
Schedules computed by our Greedy algorithm can be represented by binary trees on the jobs of the
problem instances (see Section~\ref{sec:greedy-tree} details). The binary tree ratio is the maximum ratio between a
job and its successor in the tree. We will show that our Greedy algorithm solves an instance to optimality if its binary
tree ratio is at most $2$. On the other hand, we prove that there are instances
with binary tree ratio strictly larger but arbitrarily close to $2$ that render the problem NP-hard.
A binary tree ratio of $2$ is hence the cut-off point that separates hard from easy instances.

\mypara{Other Related Works}
Since a few decades the scheduling community has been very interested in producing robust schedules that can react to changes in job characteristics.  For example in \cite{HonkompMockus:97:Robust-scheduling} a model is studied where the processing times can vary, and a schedule has to be produced with good objective value even under these variations. For more information see the survey \cite{BriandLa:07:A-robust-approach} as well as a recent PhD thesis and references therein \cite{Wilson:16:Robust-scheduling}.

Vestal \cite{ves07} introduced the \emph{mixed-criticality framework}, where the execution of lower critical jobs can be canceled in order to grant high criticality jobs the necessary amount of resources.  Applications are mostly embedded systems.  In a communication system, jobs represent messages, and safety-critical messages have to co-exist with less critical ones that are not subject to hard constraints.  For instance, the IEC 61508 standard defines four Safety Integrity Levels (SIL) (e.g. the importance of a safety-related job performing braking in a car is much higher than the importance of a job displaying the engine's temperature).
Similarly the \emph{CANaerospace} protocol specifies several criticality levels for messages, and in order to guarantee deliver times of high critical messages, the transmission of lower critical messages can be canceled \cite{Stock:14:CANaerospace-specification}.  This feature has been studied also for the \emph{FlexRay} protocol used in modern cars \cite{dvo14}.


Our model can be seen as a special case of the message transmission model described
by Hanz\'alek et al. in \cite{hts16}. They consider a single machine scheduling problem with release
times, deadlines, different criticality levels, and WCETs that depend on the task and the criticality
level. They propose a linear programming formulation of the problem and prove NP-completeness of this
more general problem. Note that our model does not consider release times and deadlines, and the WCET
of a task equals its criticality level.

Finally we would like to mention a connection with the computational problem of packing triangles in a given rectangle, which has applications in industrial cutting and storage. The later has been shown to be NP-hard~\cite{Chou:16:NP-Hard-Triangle}, while our paper shows that the problem is already hard in the particular case of right triangles that can only be vertically translated and not rotated.



\mypara{Outline}
In Section~\ref{sec:greedy} we consider our Greedy algorithms. Then,
in Section~\ref{sec:NPhardness}, we show that solving the problem for instances with binary tree ratio
strictly larger than $2$ is strongly NP-hard. Last, we provide a quasi polynomial
time approximation scheme (QPTAS) in Section~\ref{sec:QPTAS}.

\section{The Greedy algorithm}
\label{sec:greedy}




We propose a polynomial time approximation algorithm denoted \emph{Greedy}. Recall that
jobs are sorted with respect to their criticality levels, i.e., $p_1 \ge p_2 \ge \dots \ge p_n$.

\begin{definition}[Greedy]
Job 1 starts at time $s_1=0$.  Then, every job $j=2,\ldots,n$, is placed in a largest gap (the first one in case of tie).
 If the chosen gap has length $x$ and starts at time $s_i$, then the current job $j$ is placed at $s_j=s_i+p_j$.  If $2p_j>x$, then all jobs $k$ with $s_k>s_j$ are delayed by $2p_j-x$ in order to maintain feasibility, see Figure~\ref{fig:insert}.
\end{definition}

Note that in case $2p_j>x$ the makespan increases by $2p_j-x$. Hence by choosing the largest gap, Greedy minimizes the increase of the makespan at every step.

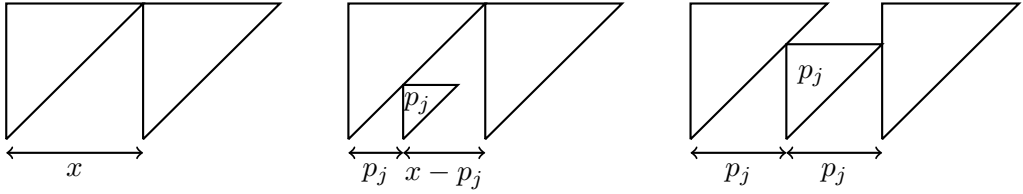
\begin{figure}[ht]
\begin{center}
\begin{tikzpicture}[scale=0.18,thick]
\joblabel{0}{10}{}
\joblabel{10}{10}{}
\joblabel{25}{10}{}
\joblabel{35}{10}{}
\joblabel{50}{10}{}
\joblabel{64}{10}{}
\joblabel{29}{4}{$p_j$}
\joblabel{57}{7}{$p_j$}
\draw[<->](0,-1) --  node[auto,swap] {$x$} (10,-1);
\draw[<->](25,-1) --  node[auto,swap] {$p_j$} (29,-1);
\draw[<->](29,-1) --  node[auto,swap] {$x-p_j$} (35,-1);
\draw[<->](50,-1) --  node[auto,swap] {$p_j$} (57,-1);
\draw[<->](57,-1) --  node[auto,swap] {$p_j$} (64,-1);
\end{tikzpicture}
\end{center}
  \caption{When Greedy inserts a job $j$ in a gap of size $x$ (left figure) it creates two gaps. One of size $p_j$ and another either of size $x-p_j$ if $x\geq 2p_j$ (center figure) or of size $p_j$ if $x<2p_j$ (right figure).}
    \label{fig:insert}
\end{figure}



\subsection{Lower bound on the optimum}







A simple lower bound can be obtained by relating gaps to jobs in the schedule
and using the fact that a gap between jobs $i$ and $j$ has size at least $\min\{p_i,p_j\}$.

\begin{lemma}  \label{lem:inner-product-lower-bound}
  Let $S=p_{\lceil n/2\rceil+1}+\ldots+p_n$ be the total processing time of the smaller half of
  the jobs, and $m=p_{(n+1)/2}$ if $n$ is odd and $m=0$ if $n$ is even. The optimal makespan OPT is at least
  \[
      m + 2S.
  \]
\end{lemma}
\begin{proof}
Consider an arbitrary feasible schedule, and let $T$ be its makespan.
To obtain a bound on $T$, we charge every gap to a job as follows:

Map every gap between jobs $i,j$ to the smaller job among them, breaking ties arbitrarily. Map the last gap
between job $j$ and the makespan to job $j$.
Now every job $j$ is the image of $0$, $1$ or $2$ gaps in this mapping. Let $a_j \in \{0, 1, 2\}$ be this number.
There are exactly $n$ gaps, hence we have $\sum_{j=1}^n a_j=n$.
Moreover, since every gap is mapped to a job of no larger size and the total gap size is $T$, we have
\[
     \sum_{j=1}^n a_j p_j \leq T.
\]
The proof follows from the fact that the left hand side is minimized when $a_j=0$ for the larger half of the
jobs, $a_j=2$ for the smaller half of the jobs and possibly $a_{(n+1)/2} = 1$ if $n$ is odd.
\end{proof}
Note that this lower bound can be very weak. Consider a $2$-job instance
with $p_1=M$ and $p_2=1$, for a large $M\geq 2$.  Then the lower bound states that the makespan is at least 2, while
the optimal makespan is $M$.

\clearpage
\subsection{Approximation ratio of Greedy}


\begin{lemma}
  The approximation ratio of Greedy is at most $1.5$.
\end{lemma}
\begin{proof}

Consider the final schedule produced by Greedy on instance $p$. We perform the following transformations:

First, we delay each job by as much as possible, while maintaining feasibility, the makespan and the job order. 
This transformation might change the order of the gap sizes, but does not modify the actual gap sizes when viewed as a multi-set.

Second, we define a \emph{truncated} instance $p'$ as follows.
For every job $j$ which starts at some $t$ followed by a gap of size $x$ with $x<p_j$, we set $p'_j=x$.  For all other jobs $j$, we set $p'_j=p_j$, see Figure~\ref{fig:truncated}. Makespan as well as feasibility of the schedule are preserved by the truncation.

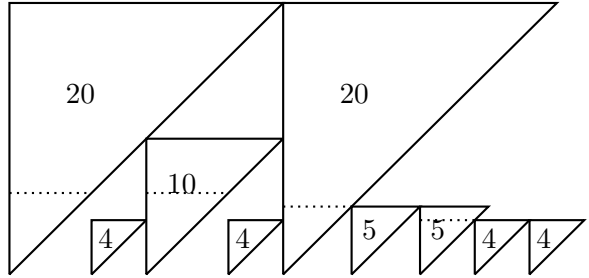
\begin{figure}
\begin{center}
\begin{tikzpicture}[scale=0.18,thick]
\job{0}{20}{0}
\job{6}{4}{0}
\job{10}{10}{0}
\job{16}{4}{0}
\job{20}{20}{0}
\job{25}{5}{0}
\job{30}{5}{0}
\job{34}{4}{0}
\job{38}{4}{0}
\draw[dotted] (0,6) -- (6,6);
\draw[dotted] (10,6) -- (16,6);
\draw[dotted] (20,5) -- (25,5);
\draw[dotted] (30,4) -- (34,4);
\end{tikzpicture}
\end{center}
\caption{Transformations on a schedule produced by Greedy: Jobs are delayed (right-shifted) and their sizes truncated (dotted lines).}
\label{fig:truncated}
\end{figure}

We assume that Greedy increased the makespan when placing the last job $n$. This assumption is without loss
of generality, since removing jobs that followed the last makespan increase only decreases the makespan of the optimal schedule,
while the makespan produced by the algorithm is preserved.

We claim that for all $i$ we have $p_n \leq p'_i < 2p_n$.  Indeed when job $n$ was placed, all gaps were of size strictly less than $2p_n$ since the insertion of job $n$ increased the makespan.  Furthermore, by induction, it can be shown that after placing job $j$, all gaps are of size at least $p_j$: By the induction hypothesis, after placing job $j-1$, all gaps were of size at least $p_{j-1}$
which is at least $p_j$, by the assumed ordering of the jobs. Then, no matter how job $j$ is placed, it is impossible that a gap of size smaller than $p_j$ is created (see also Figure~\ref{fig:insert}).

From now on, assume that $n$ is even; the proof for the odd case is similar. Let $A$ be the sum of the larger half of the
sizes among $p'_1,\ldots,p'_n$ and $B$ the sum of the smaller half. The truncation process reduces a job to the size of
gap that follows it. Therefore, the makespan of the schedule produced by Greedy on $p'$  (or on $p$)  is $A+B$.

The previous claim implies that all sizes among $p'_1,\ldots,p'_n$ are within a ratio of two, which implies $A\leq 2B$.

From Lemma~\ref{lem:inner-product-lower-bound} we have $\textrm{OPT}(p')\geq 2B$.  
Furthermore, we clearly have $\textrm{OPT}(p')\leq \textrm{OPT}(p)$. We can hence upper bound the makespan of the
schedule produced by Greedy on $p$ as
\[
  A + B \leq 3B \leq \frac32 \textrm{OPT}(p') \leq \frac32 \textrm{OPT}(p).
\]
\end{proof}

Note that this analysis did not use the fact that Greedy places jobs in the largest gap. The crucial property required
in the analysis is the fact that when the placement of a job $j$ increases the makespan, then all gaps are of size strictly
less than $2p_j$.

We were not able to determine the exact approximation factor of Greedy. In Figure~\ref{fig:lowerboundgreedy},
an instance is illustrated that shows that the approximation factor of Greedy is at least $1.05$:
On the instance with jobs of processing times $20,20,10,5,5,4,4,4,4$, Greedy
produces a schedule of makespan $42$, by placing them in the order $20,4,10,4,20,5,5,4,4$, while the optimal schedule
places the jobs in order $20,5,10,5,20,4,4,4,4$ and has makespan $40$.  This example gives the following lower bound.

\begin{lemma}
  The approximation ratio of Greedy is at least $1.05$.
\end{lemma}

We conducted a systematic search for stronger
lower bound constructions, but could only obtain tiny improvements. For example, we
found an instance consisting of $52$ jobs showing a lower bound of $101/96 > 1.052$.

\begin{figure}[ht]
\begin{tikzpicture}[scale=0.16,thick]
\job{0}{20}{0}
\job{5}{5}{0}
\job{10}{10}{0}
\job{15}{5}{0}
\job{20}{20}{0}
\job{24}{4}{0}
\job{28}{4}{0}
\job{32}{4}{0}
\job{36}{4}{0}
\end{tikzpicture}
\hfill
\begin{tikzpicture}[scale=0.16,thick]

\job{0}{20}{0}
\job{4}{4}{0}
\job{10}{10}{0}
\job{14}{4}{0}
\job{20}{20}{0}
\job{25}{5}{0}
\job{30}{5}{0}
\job{34}{4}{0}
\job{38}{4}{0}
\end{tikzpicture}
\caption{The optimal schedule (left) and the schedule produced by Greedy (right) for the lower bound instance.}
\label{fig:lowerboundgreedy}
\end{figure}
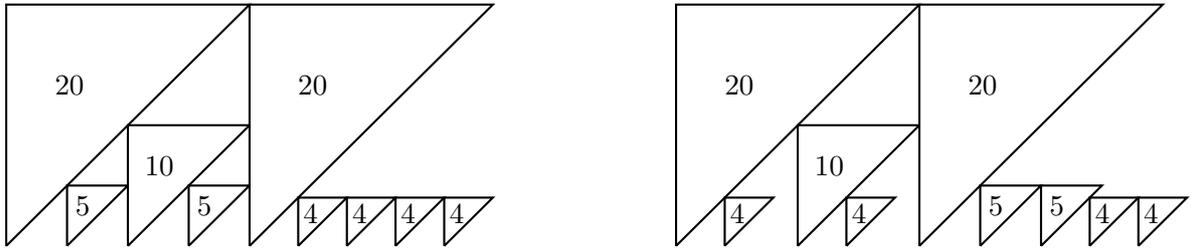



\subsection{A case where Greedy is optimal} \label{sec:greedy-tree}

\begin{theorem}
  Greedy is optimal for instances with binary tree ratio at most $2$.
\end{theorem}
\begin{proof}
We show by induction on $j$ that after placing job $j$, there are $2$ gaps of size $p_i$, for every
$\lceil j/2\rceil + 1 \leq i < j$, and either a single gap (if $j$ is odd) or $2$ gaps (if $j$ is even) of size $p_{(j+1)/2}$.  This invariant is true after placing job $1$, where there is a single gap of size $p_1$.  When $j$ is even, the job is placed in the single gap of size $p_{j/2}$.  By the assumption on the binary tree ratio we have $2p_j \geq p_{j/2}$, implying that this gap is replaced by 2 gaps of size $p_j$.
When $j$ is odd, the job is placed in one of the 2 gaps of size $p_{(j-1)/2 + 1}$, and for the same reasons as in the even case the gap is replaced by 2 gaps of size $p_j$.  In both cases the invariant is preserved.

The implication of this observation is that by the lower bound of Lemma~\ref{lem:inner-product-lower-bound}, the schedule produced by Greedy is optimal.
\end{proof}

We can relate the jobs in a tree structure as illustrated in Figure~\ref{fig:tree}. Job $1$ is the root of the tree. Then for every job $j=2,\ldots,n$ inserted into a gap assigned to job $i$, job $j$ is a descendant of job $i$. The result is a single root, connected to a binary tree, which is complete except possibly for the last level, which is left padded.  The job labels on this tree are ordered by levels.  The binary tree ratio of the instance is the maximum ratio between a job and its immediate ancestor in the tree, which was the motivation for the name of this ratio.

\begin{figure}
\tikzset{
  treenode/.style = {align=center, inner sep=0pt, text centered,
    font=\sffamily},
  arn/.style = {treenode, circle, draw=black, text width=1.5em, very thick}
}
\begin{center}
\begin{tikzpicture}[scale=0.7,->,>=stealth',level/.style={sibling distance = 25cm/#1/#1,
  level distance = 1.5cm}]
\node [arn]{1}
child{ node [arn] {2}
    child{ node [arn] {3}
            child{ node [arn] {5}
              child{ node [arn] {9}} 
              child{ node [arn] {10}}
            }
            child{ node [arn] {6}
              child{ node [arn] {11}}
              child{ node [arn] {12}}
            }
    }
    child{ node [arn] {4}
            child{ node [arn] {7}
              child{ node [arn] {13}}
              child[missing]
            }
            child{ node [arn] {8}
            }
    }
    }
;
\end{tikzpicture}
\end{center}
\caption{Tree structure of the schedule produced by Greedy on any instance of 13 jobs with binary tree ratio at most $2$.}
\label{fig:tree}
\end{figure}
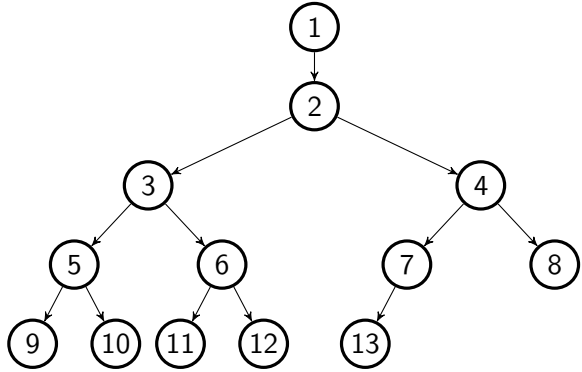

\section{NP-hardness}
\label{sec:NPhardness}

\begin{theorem}
  \textsc{TS} is strongly NP-hard for instances with binary tree ratio strictly larger than $2$.
\end{theorem}
\begin{proof}
We reduce from the strongly NP-hard numerical $3$-dimensional matching problem, see \cite[problem SP16]{GareyJohnson:79:Computers-and-intractability}.  An instance of this problem consists of integers
\[
      a_1,\ldots,a_n, b_1,\ldots,b_n, c_1,\ldots,c_n, D,
\]
with $D\geq 4$, and for all $i$:
\begin{equation}
      D/4 < a_i,b_i,c_i < D/2.
      \label{eq:NPC-promise-indiv}
\end{equation}
Furthermore, we are guaranteed that
\begin{equation}
 \sum_{i=1}^n a_i + b_i + c_i =  nD.
    \label{eq:NPC-promise-sum}
\end{equation}
The goal is to form $n$ disjoint triplets of the form $(i,j,k)$ with $a_i+b_j+c_k=D$.

Fix some arbitrary large constant $M \geq \frac{5D}{4}$.  The instance consists of $5n$ jobs.
\begin{itemize}
\item There are $n$ jobs $E$ of size $8M+5D$.
\item There are $n$ jobs $F$ of size $4M$.
\item For every $i\in\{1,\ldots,n\}$ there is a job $A_i$ of size $2M+2a_i+D$,
\item as well as a job $B_i$ of size $2M+b_i$,
\item and a job $C_i$ of size $M+c_i+D$.
\end{itemize}

We claim that the instance has a solution of makespan $n(8M+5D)$ if and only if the initial 3-Partition instance has a solution.

For the easy direction, given a solution to the 3-Partition instance we construct a schedule consisting of the concatenation for every triplet $(i,j,k)$ of the jobs $E,A_i,F,B_j,C_k$.  Straightforward verification shows that the resulting schedule has the required makespan.

For the hard direction, consider a solution to \textsc{TS} of makespan $n(8M+5D)$.  Its makespan cannot be smaller, by the presence of the $E$ jobs, that need to be scheduled every $8M+5D$ time units.  They structure the time line into $n$ blocks of equal size $8M+5D$ each.

Suppose that a block contains $k$ jobs of size $x_1\geq \ldots \geq x_k$ plus a single $E$ job.
These jobs create $k+1$ gaps in this block.  Each gap can be assigned to the smaller one among the neighboring jobs, and for an assignment we denote by $a_i\in\{0,1,2\}$ the number of assignments the job of size $x_i$ obtained.

Hence a lower bound on the block size is given by the expression $a_1x_1+\ldots+a_kx_k$ for some $\{0,1,2\}$-weights with $a_1+\ldots+a_k=k+1$.  A valid lower bound is given by $a_1 x_1 + \ldots +a_k x_k$ setting $a_i$ to $2$ for the last $\lceil k/2 \rceil$ indices (corresponding to the smaller jobs), setting $a_{k/2-1}=1$ if $k$ is even and setting $a_i=0$ to the remaining indices.  We will use this lower bound to determine the number of jobs of each type in block.

No block can host two $F$ jobs or more, since $3(4M) > 8M+5D$.  Hence every block contains exactly one $F$ job.

No block can host an $F$ and two $A$ jobs, since (placing weights $a_1=a_2=2$ for the $A$ jobs and $a_3=0$ for the $F$ job)
\[
     4 \left(2M+2\frac D 4  + D\right) = 8M + 6D.
\]
Hence every block contains exactly one $F$ and one $A$ job.

Similarly a block cannot not host an $A$ and $F$ job, together with two $B$ jobs, as setting total $a$-weight $4$ for the $B$ jobs and $1$ for the $A$ job results in the lower bound
\[
   4 \left( 2M+ \frac D 4\right) + 2M + 2\frac D 4  + D =  10M + \frac52 D.
\]
Hence every block contains exactly one $F$ one $A$ and one $B$ job.

Finally a block cannot host an $A$, $F$ and $B$ job, together with two $C$ jobs, as setting total $a$-weight $4$ for the $C$ jobs and $2$ for the $B$ job results in the lower bound
\[
   4 \left(M + \frac D 4 + D\right) + 2 \left( 2 M + 2 \frac D 4 \right)  =  8M + 6D.
\]

In conclusion every block contains an $F$ job, and $A_i$ job, a $B_j$ job and a $C_k$ job with the following lower bound for the space occupied for these jobs, using $a$-weight 2 for jobs $B_j,C_k$ and $a$-weight $1$ for job $A_i$
\[
    2M + 2a_i + D + 2(2M+b_j) + 2(M+c_k + D) = 8M + 3D + 2(a_i+b_j +c_k).
\]
Since this value cannot exceed the size of a block, namely $8M+5D$, every block corresponds to a triplet $(i,j,k)$ with $a_i+b_j+c_k\leq D$.  By the assumption (\ref{eq:NPC-promise-sum}) we have equality for every triplet, and this shows that there is a solution to the 3-Partition instance, see Figure~\ref{fig:block} for illustration.

By construction when the jobs are ordered in decreasing size, they are grouped by types, in the order $F,E,A,B,C$.  Hence the binary tree ratio is determined by the ratio between an $F$ and an $E$ job, between a $B$ and an $F$ job, and between a $C$ and an $A$ jobs.  All these fractions can be made arbitrarily close to $2$ by choosing a large enough value for $M$.
\end{proof}

\begin{figure}[ht]
 \begin{center}%
 \begin{tikzpicture}[scale=0.6,thick]
 \joblabel{0}{8.5}{E}
 \joblabel{2.2}{2.2}{A}
 \joblabel{3.2}{1}{C}
 \joblabel{4.4}{4}{F}
 \joblabel{6.4}{2}{B}
 \end{tikzpicture}
  \caption{A schematic view of a block in the optimal schedule obtained from the reduction.}%
  \label{fig:block}
 \end{center}
\end{figure}
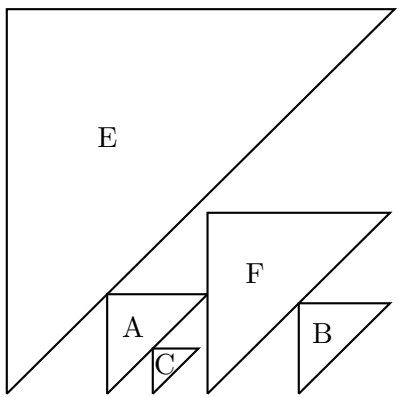




\section{A QPTAS}
\label{sec:QPTAS}

\begin{theorem}
  \textsc{TS} admits a quasi polynomial time approximation scheme.
\end{theorem}
\begin{proof}
The proof starts with a sequence of claims.

\begin{claim}
  Rounding all sizes up to the nearest power of $1 + \epsilon$ changes the optimal makespan
by at most a factor $1 + \epsilon$.
\end{claim}
\begin{proof}
  Take an optimal solution and multiply all processing times and start times by factor $1 + \epsilon$. The solution is still feasible, only the unit has changed. Now round all triangles down to the nearest power of $1 + \epsilon$  while keeping the start times fixed.
\end{proof}

\begin{claim}
We may assume that the ratio $p_1/p_n$ is at most $n/\epsilon$.
\end{claim}
\begin{proof}
  The optimal makespan $\textrm{OPT}$ is at least $p_1$. All triangle with size less than $\epsilon p_1/n$ can be put at the end. We do not need to optimize over these, as they increase the makespan by at most an $\epsilon$ factor.
\end{proof}

From now let the smallest triangle have size $p_n=1$ and the largest have size $p_1 \leq n/\epsilon$. Then $p_1 \leq \textrm{OPT} \leq n p_1 \leq n^2/\epsilon$.

\begin{claim}
We may assume there are only $\lceil \log (n/\epsilon) \rceil +1$  number of different sizes.
\end{claim}
The proof follows from the previous two claims.

\begin{claim}
Restricting start times to values from the set $P=\{0, K, 2K, ..., \lceil n^2/\epsilon \rceil K\}$ for $K = \epsilon p_1/n$ increases the optimal makespan at most by factor $1+\epsilon$.
\end{claim}
\begin{proof}
We modify the start times of the jobs, processing them from left to right. For every job, we move it to the next time in $P$, and move simultaneously all subsequent jobs by the same amount to preserve feasibility.  The value of the solution increases by at most $nK \leq \epsilon p_1 \leq \epsilon \textrm{OPT}$.
\end{proof}
Let $Z$ be the set of all different job sizes after the above rounding.
We have
\begin{align*}
  |Z| & = \lceil \log (n/\epsilon) \rceil +1
  \\
  |P| &\leq \lceil n^2 / \epsilon \rceil
\end{align*}
We design a dynamic programming scheme as follows.
Partition the set of jobs into sets $S$ and its complement $\overline S$. We want to compute the optimal schedule placing first jobs from $\overline S$ and then jobs from $S$.  Only a few parameters from the first schedule influence the possibilities of the second schedule, namely the position of the rightmost triangle of each size. Hence we describe every possible configuration by a vector from $P^Z$, which leads to a quasi polynomial number of configurations. The number of possible subsets $S$ is (in terms of size multiplicities) $O(n^Z)$ which is also quasi-polynomial.

Define $F (C, S)$ to be the optimal makespan of a schedule, placing first $\overline S$ with a configuration $C$ and placing then jobs from $S$.
The goal is to compute $F(e, \{1,\ldots,n\})$, where $e$ is the empty configuration, which for technical reasons assigns to each size $x\in Z$ the starting time $-x$.

The basis cases consist of $F(C,\emptyset)$, where $C$ ranges over all configurations from $P^Z$, some of them might be infeasible.
Then $F (C, S)$
 can be computed from $F (C+j, S-j)$ where $S-j$ is $S$ minus one triangle $j\in S$ and $C+j$ is obtained from $C$ by adding the triangle $j$ to the right of $C$ and placing it as early as possible, i.e.\ at time $\max_{x\in Z} C_x + x$.
 The number of choices for $j$ (in terms of size) is at most $|Z|$. Hence, we need at most $|Z|$ look-ups to compute one value.
\end{proof}

\clearpage
\section{Final remarks}

\shapepar{{3}{0}b{0}\\{0}t{0}{6}\\{6}e{0}}
We introduced a new scheduling problem, motivated by mixed criticality. The novelty lies in its combinatorial structure which defines the contribution to the makespan of each job in a non-local manner. We showed that the problem is strongly NP-hard, thus ruling out the existence of a fully polynomial time approximation schedule.  In addition we provided a quasi polynomial time approximation scheme, ruling out APX-hardness.  Furthermore, we introduced a greedy algorithm for this problem, but still do not understand well its approximation ratio. Closing the gap between the lower bound of $1.05$ and the upper bound of $1.5$ is the main question left open by this paper.
We would like to thank Marek Chrobak and Neil Olver for helpful discussions.
This work is partially supported by PHC VAN GOGH 2015
PROJET 33669TC, the grants FONDECYT 11140566 as well as ANR-15-CE40-0015.

\bibliographystyle{plain}
\bibliography{references}

\end{document}